\documentclass[a4paper, twoside, reqno, 12pt, dvips]{amsart}
\RequirePackage{fix-cm} 

\usepackage{fixltx2e}     

\usepackage{eucal}
\usepackage{xspace}
\usepackage{amsgen}
\usepackage{amssymb}
\usepackage{amsmath}
\usepackage{amsfonts}
\usepackage{MnSymbol}
\usepackage[nice]{nicefrac}
\usepackage{mathrsfs, dsfont}

\usepackage{a4wide}

\usepackage{acronym}

\usepackage[section]{placeins}
\usepackage{indentfirst}
\usepackage{graphicx}
\usepackage{psfrag}

\usepackage[usenames, dvipsnames, pdf]{pstricks}
\usepackage{epsfig}
\usepackage{pst-grad} 
\usepackage{pst-plot} 

\usepackage{subfigure}

\usepackage{longtable}

\usepackage[french,english]{babel}
\usepackage[latin1]{inputenc}

\usepackage{color}
\definecolor{oneblue}{rgb}{0,0.0,0.75}
\definecolor{darkgrey}{rgb}{0.273, 0.281, 0.30}

\usepackage[colorlinks,
            urlcolor=oneblue,
            linkcolor=oneblue,
            citecolor=oneblue,
            bookmarksopen=false,
            pagebackref]{hyperref}

\usepackage[compact]{titlesec}

\titleformat{\section}{\normalfont\Large\bfseries\sffamily\center\color{darkgrey}}{\thesection.}{0.5em}{}{}
\titleformat{\subsection}{\normalfont\large\bfseries\sffamily\color{darkgrey}}{\thesubsection.}{0.4em}{}{}
\titleformat{\subsubsection}{\normalfont\normalsize\bfseries\sffamily\color{darkgrey}}{\thesubsubsection.}{0.3em}{}{}

\titlespacing*{\section}{1.0em}{1.0em}{0.8em}[0em]
\titlespacing*{\subsection}{1.0em}{1.0em}{0.8em}[0em]
\titlespacing*{\subsubsection}{1.0em}{0.7em}{0.6em}[0em]

\usepackage{titletoc}

\contentsmargin{0.0em}
\dottedcontents{section}[2.5em]{\addvspace{0.45em}\bfseries}{1.0em}{0pc}
\dottedcontents{subsection}[4.5em]{}{2.6em}{0pc}
\dottedcontents{subsubsection}[5.5em]{}{3.0em}{0pc}

\usepackage{fancyhdr}
\usepackage{lastpage}

\newcommand*\Title{Generation of water waves by moving bottom disturbances}
\newcommand*\Authors{H.~Nersisyan, D.~Dutykh \& E.~Zuazua}

\numberwithin{equation}{section}

\newtheorem{remark}{Remark}
\newtheorem{theorem}{Theorem}

\newcommand{\R}{\mathbb{R}}
\newcommand{\T}{\mathbb{T}}
\newcommand{\Z}{\mathbb{Z}}
\newcommand{\I}{\mathcal{I}}
\newcommand{\N}{\mathcal{N}}
\newcommand{\ui}{\mathrm{i}}

\newcommand{\eps}{\varepsilon}

\newcommand{\ud}{\mathrm{d}\/}
\newcommand{\teta}{\tilde{\eta}}

\DeclareMathOperator{\sech}{sech}

\newcommand{\supp}{\mathop{\mathrm{supp}}\nolimits}

\newcommand{\od}[2]{\frac{\mathrm{d}\,#1}{\mathrm{d}\/#2}}

\begin{document}

\title[\Title]{Generation of two-dimensional water waves by moving bottom disturbances}

\author[H.~Nersisyan]{Hayk Nersisyan}
\address{BCAM - The Basque Center for Applied Mathematics, Alameda Mazarredo 14, 48009 Bilbao, Basque Country -- Spain}
\email{hnersisyan@bcamath.org}
\urladdr{http://www.bcamath.org/en/people/nersisyan/}

\author[D.~Dutykh]{Denys Dutykh$^*$}
\address{University College Dublin, School of Mathematical Sciences, Belfield, Dublin 4, Ireland \and LAMA, UMR 5127 CNRS, Universit\'e de Savoie, Campus Scientifique, 73376 Le Bourget-du-Lac Cedex, France}
\email{Denys.Dutykh@univ-savoie.fr}
\urladdr{http://www.denys-dutykh.com/}
\thanks{$^*$ Corresponding author}

\author[E.~Zuazua]{Enrique Zuazua}
\address{BCAM - The Basque Center for Applied Mathematics, Alameda Mazarredo 14, 48009 Bilbao, Basque Country -- Spain \and Ikerbasque, Basque Foundation for Science  Alameda Urquijo 36-5, Plaza Bizkaia 48011, Bilbao, Basque Country, Spain}
\email{zuazua@bcamath.org}
\urladdr{http://www.bcamath.org/en/people/zuazua/}

\begin{abstract}
We investigate the potential and limitations of the wave generation by disturbances moving at the bottom. More precisely, we assume that the wavemaker is composed of an underwater object of a given shape which can be displaced according to a prescribed trajectory. We address the practical question of computing the wavemaker shape and  trajectory generating a wave with prescribed characteristics. For the sake of simplicity we model the hydrodynamics by a generalized forced Benjamin--Bona--Mahony (BBM) equation. This practical problem is reformulated as a constrained nonlinear optimization problem. Additional constraints are imposed in order to fulfill various practical design requirements. Finally, we present some numerical results in order to demonstrate the feasibility and performance of the proposed methodology.
\end{abstract}


\maketitle

\tableofcontents

\section{Introduction}

The problem of wave generation is complex and has many practical applications. On the scale of a laboratory wave tank a traditional wavemaker is composed of numerous paddles attached to a vertical wall and which can move independently according to a prescribed program. These wavemakers have been successfully used to conduct laboratory experiments at least since late 60's \cite{Feir1967, Lake1977}.

In this study we investigate theoretically and numerically the potential for practical applications of a different kind of wave making devices. Namely, the mechanism considered hereinbelow is composed mainly of an underwater object which can be displaced along a portion of the bottom with the prescribed trajectory. In mathematical terms, we study the wave excitation problem by moving forcing at the bottom. Similar processes are known in physics under the name of autoresonance phenomena, thoroughly studied by L.~\textsc{Friedland} and his collaborators \cite{Friedland1998a, Friedland1998, Friedland2003}.

Recently, this type of wavemakers has found an interesting application to the man-made surfing facilities \cite{InstantSport2012}. The device was proven to be successful to generate high quality waves for surfing far from the Oceans. Our main goal consists in providing some elements of the modelling and theoretical analysis of this process. The second objective of this study is to provide an efficient computational procedure to determine the underwater object shape and trajectory to generate a prescribed wave profile in a given portion of the wave tank.

\begin{figure}
  \centering
  \includegraphics[width=0.99\textwidth]{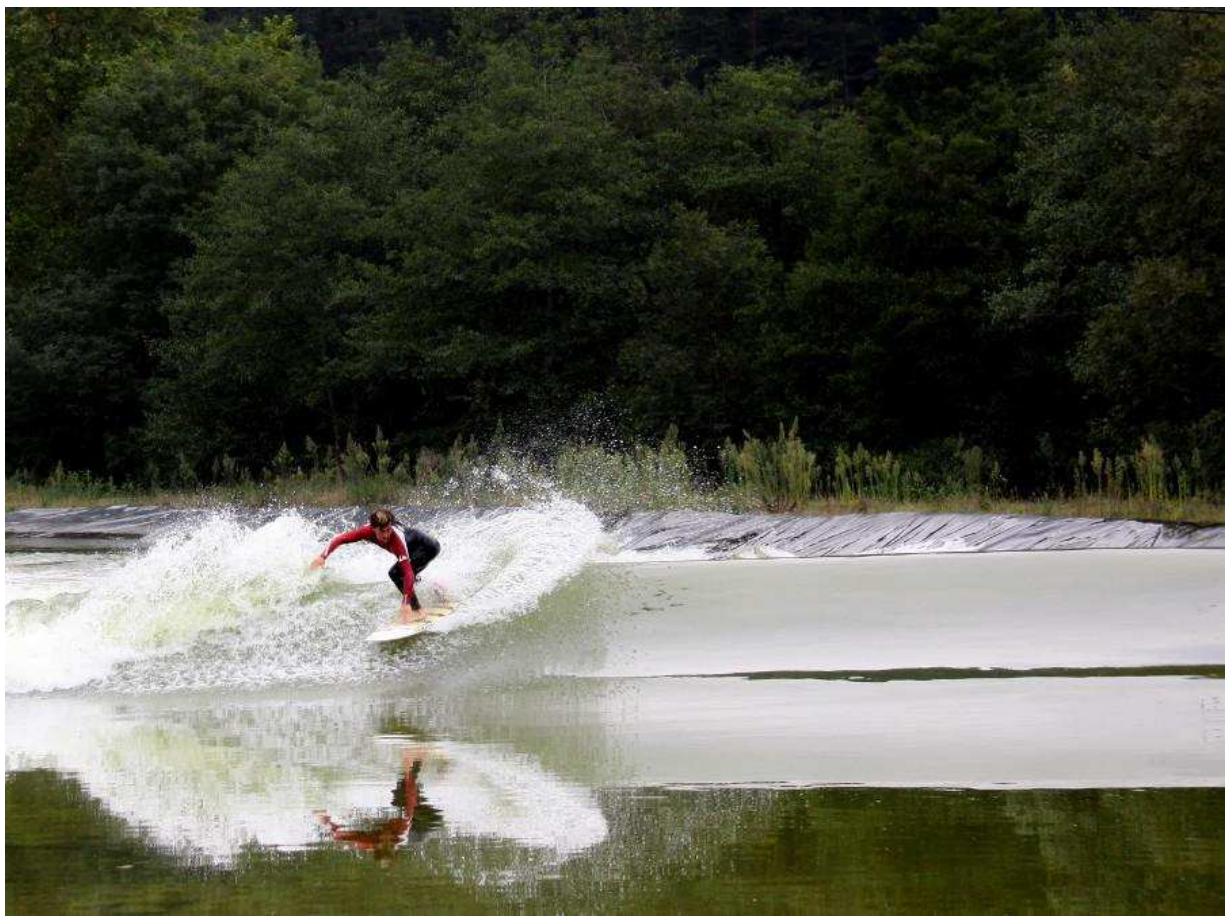}
  \caption{\small\em An artificial wave generated in a pool by an underwater wave making device. \copyright\ \url{http://www.wavegarden.com/}}
  \label{fig:surf}
\end{figure}

The problem of wave generation by moving bottom has been particularly studied in the context of tsunami waves genesis. These extreme waves are caused by sea bed displacements due to an underwater earthquake \cite{Hammack, bradd, Nosov2001, Dutykh2007b, Dutykh2010a} or a submarine landslide \cite{Ward, Okal2003, Watts2003, Beisel2011}. It is mainly the vertical bottom motion which contributes most to the tsunami generation by earthquakes, even if some effort has been made to take into account also for the horizontal displacement components \cite{Tanioka1996, Todo, todo2, Nosov2009, Dutykh2010d}. On the other hand, the wave making mechanism studied here involves \emph{only} the horizontal motion. Consequently, the methods and known results from the tsunami wave community cannot be directly transposed to this problem.

The wave propagation takes place in a shallow channel, so the long wave assumption can be adopted \cite{Ursell1953, Craig1994}. However, weak dispersive and weak nonlinear effects should be included since the resulting wave observed in experiments has some common characteristics with a solitary wave. Consequently, as the base model we choose the classical Boussinesq system derived by D.H.~\textsc{Peregrine} (1967) \cite{Peregrine1967} and generalized later by T.~\textsc{Wu} (1987) \cite{Wu1987}, who included the time-dependent bathymetry effects. In order to simplify further the problem, we assume the wave propagation to be unidirectional and, hence, we derive a generalized forced Benjamin--Bona--Mahony (BBM) equation \cite{bona}. This equation is then discretized with a high order finite volume method \cite{Benkhaldoun2008, Dutykh2011e, ChazelLannes2010, Dutykh2010e}. Finally, the trajectory and the shape of the underwater wavemaker are optimized in order to minimize a cost-function under some practical constraints.

From the mathematical point of view our formulation can be seen as a controllability problem for the forced BBM equation \cite{Pontryagin1987}. There is an extensive literature on the controllability of dispersive wave equations such as KdV \cite{Boussinesq1877, KdV}, BBM \cite{bona} and some Boussinesq-type systems \cite{Bona2004}.
 For instance, the controllability and stabilizability of the KdV equation
\begin{equation*}
  u_t + u_{xxx} + u_x + uu_x = 0, \qquad x \in [0,L], t>0
\end{equation*}
was addressed by L.~\textsc{Russell} \& B.-Y.~\textsc{Zhang} (1993) \cite{Russell1993} for periodic boundary conditions with an internal control. The boundary control was investigated by the same authors later \cite{Russell1995}. The controllability of the KdV equation with Dirichlet boundary conditions was studied in the following papers \cite{Russell1996, Rosier1997, Zhang1999, Perla-Menzala2002, Coron2004, Rosier2004,   Cerpa2007, Glass2008, Cerpa2009, Glass2010, Laurent2010}, the list of references not being exhaustive.
On the other hand, \textsc{Rosier} \& \textsc{Zhang} (2012) \cite{Rosier2012} proved the Unique Continuation Property (UCP) for the solution of BBM equation on a one dimensional torus $\T:= \R/ 2\pi\Z$   with small enough initial
data from $H^1(\T)$  with nonnegative mean values:
\begin{align}
  u_t - u_{txx} + u_x + uu_x = 0,  \quad  u(0)=u_0, \label{eq:roBBM}\\
 \int _ \T u_0(x) \ud x  \ge 0, \quad || u_0||_\infty <3,
\end{align}
i.e. for any open nonempty set $\omega\subset \T$ the only solution of \eqref{eq:roBBM} with
\begin{equation*}
  u(x,t) = 0 \text{ for } (x,t)\in \omega\times (0,T)
\end{equation*}
is the trivial solution $u = 0$. They also considered the following control problem, with a moving control,
\begin{align*}
  u_t - u_{txx} + u_x + uu_x = a(x + ct)h(x, t),
\end{align*}
where $a\in C^\infty$ is given and $h(x,t)$ is the control and proved local  exact controllability in $H^s(\T)$ for any $s \ge 0$ and global exact  controllability in $H^s(\T)$ for any $s \ge 1$. A necessary and sufficient algebraic condition for approximate controllability of the BBM equation with homogeneous Dirichlet boundary conditions was given in \cite{Adames2008}. The controllability of the linearized BBM and KdV equations was studied in \cite{Rosier2000, Micu2001, Zhang2003}. The controllability of a family of Boussinesq equations was studied theoretically as well, see \cite{Zhang2009}. The  controllability of the heat and wave equations with a moving boundary was also recently addressed in \cite{Touboul2012}.

The present paper is organized as follows. In Section \ref{sec:model} we derive the governing equation. Then, this model is analysed mathematically in Section \ref{sec:well}. The optimization is discussed in Section \ref{sec:opt}.
The results of some numerical simulations are presented in Section \ref{sec:num}. Finally, in Section \ref{sec:concl} we outline the main conclusions of this study.

\section{The mathematical model}\label{sec:model}

Consider an ideal incompressible fluid of constant density in a two-dimensional domain. The horizontal independent variable is denoted by $x$ and the upward vertical one by $y$. The origin of the Cartesian coordinate system is chosen such that the line $y=0$ corresponds to the still water level. The fluid is bounded below by an impermeable bottom at $y = -h (x,t)$ and above by an impermeable free surface at $y = \eta (x,t)$. We assume that the total depth $H(x, t) \equiv h (x,t) + \eta (x,t)$ remains positive $H (x,t) \geqslant H_0>0$ at all times $t$. The sketch of the physical domain is shown in Figure \ref{fig:sketch}. The depth-averaged horizontal velocity is denoted by $u(x,t)$ and the gravity acceleration by $g$. The fluid layer has the uniform depth $d$ everywhere, which is perturbed only by a localized object, which can move along the bottom:
\begin{equation}\label{eq:topogr}
  h(x,t) = d - \zeta(x,t), \quad \zeta(x,t) = \zeta_0(x - x_0(t)),
\end{equation}
where the function $\zeta_0(x)$ has a compact support and $x = x_0(t)$ is the trajectory of its barycenter. The meaning of the segment $\I = [a, b]$ is explained in Section \ref{sec:opt}.

\begin{figure}
\centering
\scalebox{1} 
{
\begin{pspicture}(0,-2.32)(11.381875,2.32)
\definecolor{color22g}{rgb}{0.8549019607843137,0.6078431372549019,0.5725490196078431}
\definecolor{color22f}{rgb}{0.4588235294117647,0.15294117647058825,0.10980392156862745}
\definecolor{color87b}{rgb}{0.40784313725490196,0.1411764705882353,0.1411764705882353}
\definecolor{color132}{rgb}{0.13333333333333333,0.18823529411764706,0.8549019607843137}
\psline[linewidth=0.03cm,linestyle=dashed,dash=0.16cm 0.16cm,arrowsize=0.05291667cm 2.0,arrowlength=1.4,arrowinset=0.4]{->}(0.0,0.88)(11.3,0.88)
\psline[linewidth=0.04cm,linestyle=dashed,dash=0.16cm 0.16cm,arrowsize=0.05291667cm 2.0,arrowlength=1.4,arrowinset=0.4]{<-}(5.4,2.3)(5.38,-2.3)
\psframe[linewidth=0.04,dimen=outer,fillstyle=gradient,gradlines=2000,gradbegin=color22g,gradend=color22f,gradmidpoint=1.0](10.78,-1.56)(0.46,-1.88)
\usefont{T1}{ptm}{m}{n}
\rput(11.081407,0.585){$x$}
\usefont{T1}{ptm}{m}{n}
\rput(5.1314063,2.105){$y$}
\psarc[linewidth=0.044,linestyle=dashed,dash=0.16cm 0.16cm,fillcolor=color87b](3.15,-1.55){0.39}{0.0}{180.0}
\psline[linewidth=0.044,linestyle=dashed,dash=0.16cm 0.16cm](3.54,-1.55)(2.76,-1.55)
\psline[linewidth=0.04cm,linestyle=dashed,dash=0.16cm 0.16cm,arrowsize=0.05291667cm 2.0,arrowlength=1.4,arrowinset=0.4]{->}(3.48,-1.34)(4.3,-1.34)
\psarc[linewidth=0.044,fillstyle=solid,fillcolor=color87b](7.29,-1.55){0.39}{0.0}{180.0}
\psline[linewidth=0.044](7.68,-1.55)(6.9,-1.55)
\pscustom[linewidth=0.04,linecolor=color132]
{
\newpath
\moveto(0.34,0.7)
\lineto(0.61,0.8)
\curveto(0.745,0.85)(0.945,0.935)(1.01,0.97)
\curveto(1.075,1.005)(1.295,1.035)(1.45,1.03)
\curveto(1.605,1.025)(1.885,0.97)(2.01,0.92)
\curveto(2.135,0.87)(2.42,0.805)(2.58,0.79)
\curveto(2.74,0.775)(3.02,0.795)(3.14,0.83)
\curveto(3.26,0.865)(3.575,0.925)(3.77,0.95)
\curveto(3.965,0.975)(4.355,0.96)(4.55,0.92)
\curveto(4.745,0.88)(5.12,0.815)(5.3,0.79)
\curveto(5.48,0.765)(5.76,0.765)(5.86,0.79)
\curveto(5.96,0.815)(6.135,0.885)(6.21,0.93)
\curveto(6.285,0.975)(6.445,1.1)(6.53,1.18)
\curveto(6.615,1.26)(6.795,1.39)(6.89,1.44)
\curveto(6.985,1.49)(7.145,1.58)(7.21,1.62)
\curveto(7.275,1.66)(7.41,1.69)(7.48,1.68)
\curveto(7.55,1.67)(7.655,1.605)(7.69,1.55)
\curveto(7.725,1.495)(7.795,1.375)(7.83,1.31)
\curveto(7.865,1.245)(7.93,1.115)(7.96,1.05)
\curveto(7.99,0.985)(8.075,0.875)(8.13,0.83)
\curveto(8.185,0.785)(8.31,0.775)(8.38,0.81)
\curveto(8.45,0.845)(8.6,0.92)(8.68,0.96)
\curveto(8.76,1.0)(8.905,1.01)(8.97,0.98)
\curveto(9.035,0.95)(9.175,0.89)(9.25,0.86)
\curveto(9.325,0.83)(9.465,0.82)(9.53,0.84)
\curveto(9.595,0.86)(9.74,0.915)(9.82,0.95)
\curveto(9.9,0.985)(10.065,1.025)(10.15,1.03)
\curveto(10.235,1.035)(10.395,1.02)(10.47,1.0)
\curveto(10.545,0.98)(10.675,0.965)(10.84,0.98)
}
\psline[linewidth=0.03cm,linestyle=dashed,dash=0.16cm 0.16cm,arrowsize=0.05291667cm 2.0,arrowlength=1.4,arrowinset=0.4]{<->}(1.28,0.88)(1.28,-1.56)
\usefont{T1}{ptm}{m}{n}
\rput(1.5314063,-0.455){$d$}
\usefont{T1}{ptm}{m}{n}
\rput(3.9914062,-1.095){$x_0(t)$}
\usefont{T1}{ptm}{m}{n}
\rput(2.1114063,1.865){$y = \eta(x,t)$}
\psline[linewidth=0.04cm,linestyle=dashed,dash=0.16cm 0.16cm,arrowsize=0.05291667cm 2.0,arrowlength=1.4,arrowinset=0.4]{->}(3.0,1.6)(3.76,1.0)
\psline[linewidth=0.04cm,linestyle=dashed,dash=0.16cm 0.16cm](6.08,0.86)(6.08,0.52)
\psline[linewidth=0.04cm,linestyle=dashed,dash=0.16cm 0.16cm](8.04,0.9)(8.04,0.56)
\psline[linewidth=0.04cm,linestyle=dashed,dash=0.16cm 0.16cm,arrowsize=0.05291667cm 2.0,arrowlength=1.4,arrowinset=0.4]{<->}(6.1,0.62)(8.04,0.62)
\usefont{T1}{ptm}{m}{n}
\rput(7.161406,0.305){$\I = [a, b]$}
\end{pspicture} 
}
\caption{\small\em Sketch of the physical domain with an underwater object moving along the bottom.}
\label{fig:sketch}
\end{figure}
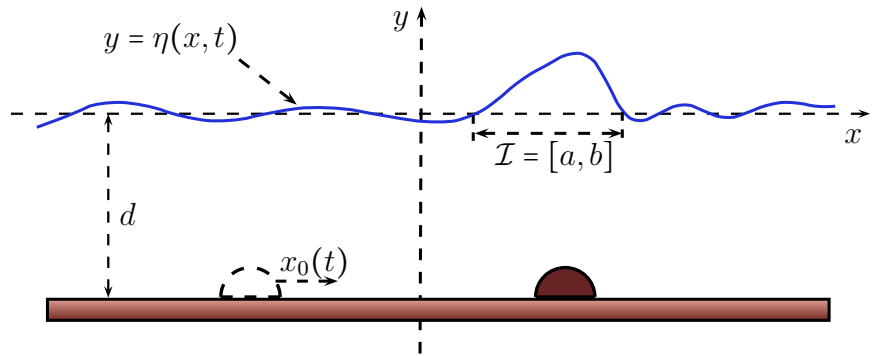

In 1987 T.~\textsc{Wu} \cite{Wu1987} derived the following Boussinesq-type system to study the generation of solitary waves by moving disturbances:
\begin{eqnarray}
\eta_t +\left(\left(h+\eta\right)u\right)_x &=& -h_t, \label{eq:Bouss1}\\
u_t + uu_x + g\eta_x &=& \frac{1}{2} h\left(h_t + (hu)_x\right)_{xt} - \frac{1}{6}h^2u_{xxt}. \label{eq:Bouss2}
\end{eqnarray}
This system represents a further generalization of the classical Boussinesq equations derived by D.H.~\textsc{Peregrine} (1967) \cite{Peregrine1967} for the case of the moving bottom $h(x,t)$. In our work we take the system \eqref{eq:Bouss1}, \eqref{eq:Bouss2} as the starting point. In order to simplify it further, we will switch to dimensionless variables (denoted with primes):
\begin{align*} 
x' = \frac{x}{l}, \quad \eta' = \frac{\eta}{a}, \quad h' = \frac{h}{d},  \quad u' = \frac{u}{\frac{ga}{\sqrt{gd}}}, \quad t' = \frac{t}{\frac{\ell}{\sqrt{gd}}},  \text{ and }  \zeta' = \frac{\zeta}{a}, 
\end{align*}
where $a$ and $\ell$ are the characteristic wave amplitude and wavelength correspondingly. We can compose three important dimensionless numbers which characterize the Boussinesq regime:
\begin{equation*} 
 \eps := \frac{a}{d} \ll 1, \quad \mu^2 := \left(\frac{d}{\ell}\right)^2 \ll 1, \text{ and } S := \frac{\eps}{\mu^2} = O(1),
\end{equation*}
where $S$ is the so-called Stokes-Ursell number \cite{Ursell1953}, which measures the relative importance of dispersive and nonlinear effects. In unscaled variables the Peregrine-Wu system takes the following form (for simplicity we drop out the primes):
\begin{eqnarray}\label{eq:bdimless1}
  \eta_t +\left(\left(h+\eps\eta\right)u\right)_x &=& -h_t, \\
  u_t + \eps uu_x + \eta_x &=& \frac{\mu^2}{2}h\left(h_t + (hu)_x\right)_{xt} - \frac{\mu^2}{6} h^2u_{xxt} \label{eq:bdimless2}.
\end{eqnarray}
To simplify the problem, we will reduce the Boussinesq system \eqref{eq:Bouss1}, \eqref{eq:Bouss2} to the unidirectional wave propagation. For instance, in the original work of T.~\textsc{Wu} \cite{Wu1987} a similar reduction to a forced KdV is also performed. However, the resulting model in our work will be of the BBM-type \cite{bona}, since it possesses better numerical stability properties.

The reduction to the BBM equation can be done in the following way \cite{Johnson2004}. The horizontal velocity $u$ can be approximatively represented in unscaled variables as
\begin{equation}\label{eq:unidirect}
  u = \eta + \eps P + \mu^2 Q + O(\eps^2 + \eps\mu^2 + \mu^4),
\end{equation}
where $P(x,t)$ and $Q(x,t)$ are some functions to be determined. The sign $+$ in front of $\eta$ means that we consider the waves moving in the rightward direction. Substituting the representation \eqref{eq:unidirect} into unscaled Boussinesq equations \eqref{eq:bdimless1}, \eqref{eq:bdimless2} yields two equivalent relations:
\begin{eqnarray}
\eta_t + \eta_x + \eps P_x + \mu^2 Q_x + 2\eps\eta\eta_x - \eps(\zeta\eta)_x  - \eps \zeta_t= O(\eps^2 + \eps\mu^2 + \mu^4), \label{eq:Boussunidir1}\\
\eta_t + \eta_x + \eps P_t + \mu^2 Q_t + \eps\eta\eta_x = \frac{\mu^2}{2} h\left(-\zeta_t + \left(h\eta\right)_x \right)_{xt} - \frac{\mu^2}{6} h^2 \eta _{xxt} + O(\eps^2 + \eps\mu^2 + \mu^4). \label{eq:Boussunidir2}
\end{eqnarray}
By subtracting \eqref{eq:Boussunidir1} from \eqref{eq:Boussunidir2} we obtain the following compatibility condition:
\begin{multline}\label{eq:PQ}
  \eps\left(P_x - P_t\right) + \mu^2\left(Q_x - Q_t\right) + \eps\eta\eta_x - \eps\left(\zeta\eta\right)_x   = \\
  \eps \zeta_t -\frac{\mu^2}{2}\left(-\zeta_t + \eta_x \right)_{xt} + \frac{\mu^2}{6}h^2\eta_{xxt} + O(\eps^2 + \eps\mu^2 + \mu^4).
\end{multline}
For the right-going waves we have the following identities:
\begin{equation*}
  P_t = -P_x + O(\eps), \qquad Q_t = -Q_x + O(\eps).
\end{equation*}
Finally, the unknown functions $P_x$ and $Q_x$ can be chosen to satisfy asymptotically the compatibility condition \eqref{eq:PQ}, which yields
\begin{equation*}
  2P_x = (\zeta\eta)_x - \eta\eta_x  + \zeta_t, \qquad 2Q_x = \frac{1}{2}\zeta_{xtt} -\frac{1}{3}\eta_{xxt},
\end{equation*}
where we used also the analytical representation of $h(x,t) = 1 - \eps\zeta(x,t)$. The BBM equation in unscaled variables can be now easily obtained by substituting expressions for $P_x$ and $Q_x$ into equation \eqref{eq:Boussunidir1}:
\begin{equation*}
  \eta_t + \eta_x + \frac{\eps}{2}\left((\zeta\eta)_x - \eta\eta_x + \zeta_t \right) + \frac{\mu^2}{2}\left(\frac{1}{2}\zeta_{xtt} - \frac{1}{3}\eta_{xxt}\right) + 2\eps\eta\eta_x - \eps(\zeta\eta)_x = \eps  \zeta_t.
\end{equation*}
Turning back to physical variables, the generalized forced BBM (gBBM) equation takes the following form:
\begin{equation*}
  \eta_t + \left(\sqrt{gd}\eta + \sqrt{\frac{g}{d}}\left(\frac{3}{4}\eta^2 - \frac{1}{2}\zeta\eta\right)\right)_x - \frac{d^2}{6}\eta_{xxt} = -\frac{1}{4}\frac{d^2}{\sqrt{gd}}\zeta_{xtt} +\frac{1}{2} \zeta_t.
\end{equation*}
In subsequent sections we will use this equation to model wave-bottom interaction. However, in order to simplify the notations, we will introduce a new set of dimensionless variables, where all the lengthes are unscaled with the water depth $d$, velocities with $\sqrt{gd}$ and the time variable with $\displaystyle{\sqrt{d/g}}$. In this scaling the gBBM equation reads:
\begin{equation}\label{eq:gbbm}
  \eta_t + \left(\eta + \frac{3}{4}\eta^2 - \frac{1}{2}\zeta\eta\right)_x - \frac{1}{6}\eta_{xxt} = -\frac{1}{4}\zeta_{xtt} +\frac{1}{2} \zeta_t.
\end{equation}
Recall that  here  $\eta$ is the unknown free surface elevation, and $\zeta$ is a given function, which is the topography of the moving body defined by \eqref{eq:topogr}.
The last equation has to be completed by appropriate initial and boundary conditions (when posed on a finite or semi-infinite domain):
\begin{equation}\label{eq:gbbm2}
  \eta(x,0) = \eta_0(x), \quad x\in \R.
\end{equation}

The method that we used to get the model \eqref{eq:gbbm} is  known in the literature. For instance, in  \cite{Bona2005}, by  similar arguments, the model  waves generation by a moving boundary was obtained.  The  BBM system was also justified by laboratory experiments \cite{Bona1981}.

\section{Well-posedness of the gBBM equation}\label{sec:well}

In this section we give a proof of the well-posedness of the gBBM equation \eqref{eq:gbbm} in the Sobolev spaces $H^s := H^s(\R)$. First, we have the following result.

\begin{theorem}\label{wellpos}
For any $\zeta \in C^2([0,\infty),H^{s}) \cap C([0,\infty),H^{[s]+1})$ and $\eta_0 \in H^s$, $s\ge0$ the problem \eqref{eq:gbbm}, \eqref{eq:gbbm2} admits a unique solution $\eta\in C([0,\infty),H^{s})$.
\end{theorem}

\begin{proof}
\textbf{Uniqueness.} Let us assume that for some given functions $\zeta$ and $\eta_0$ our problem \eqref{eq:gbbm}, \eqref{eq:gbbm2} admits two different solutions $\eta_1$ and $\eta_2$. The difference $\teta := \eta_1 - \eta_2$ satisfies the following initial-value problem:
\begin{equation}\label{eq:uniqueness}
  \teta_t + \left(\teta + \frac{3}{4}\teta(\eta_1 + \eta_2) - \frac{1}{2}\zeta\teta\right)_x - \frac{1}{6}\teta_{xxt} = 0, \qquad \teta (x,0) = 0.
\end{equation}
Using Fourier transformation, we can rewrite \eqref{eq:uniqueness} in the form
\begin{equation}\label{eq:uniqinfourier}
  \ui \teta_t= \varphi(D_x) \left(\teta + \frac{3}{4}\teta(\eta_1 + \eta_2) - \frac{1}{2}\zeta\teta\right),
\end{equation}
where $\varphi(D_x)$ is defined by $\widehat{\varphi(D_x)v}(\xi) := \frac {\xi}{ 1+1/6 \xi^2} \hat v (\xi).$ Clearly,  \eqref{eq:uniqinfourier} implies 
\begin{equation*}
\teta(x,t) = -\ui \int\limits_0^t\varphi(D_x) \left(\teta + \frac{3}{4}\teta(\eta_1 + \eta_2) - \frac{1}{2}\zeta\teta\right) \ud t.
\end{equation*}
From the inequality 
\begin{equation*}  
\| \varphi(D_x) (uv) \|_0 \le C\|u\|_0\|v\|_0
\end{equation*}
it follows that
\begin{equation}\label{eq:uniqinintegralform}
  \sup _{t\in [0,T]} \|\teta (x,t )\|_0 \le    C T  \left( 1+\|\eta_1\|_0 + \|\eta_2\|_0 + \|\zeta \|_0\right) \|\teta\|_0.
\end{equation}
Thus, the application of the Gronwall inequality yields $\teta = 0$. \qed

\textbf{Existence.} For any fixed time horizon $T>0$ let us show that our problem \eqref{eq:gbbm}, \eqref{eq:gbbm2} has a solution $\eta\in C([0,T),H^{s})$. J.~\textsc{Bona} \& N.~\textsc{Tzvetkov} (2009) \cite{Bona2009} proved that for any given initial data $\eta_0 \in H^s$ the following BBM equation
\begin{equation*}
  u_t + u_x + uu_x - u_{xxt} = 0, \qquad u(x,0) = u_0(x),
\end{equation*}
admits a unique solution $u \in C( [0,\infty), H^s)$ (for $s<0$ they proved that the system is ill-posed). Using a scaling argument, we have also the well-posedness of the same equation with some positive coefficients:
\begin{equation}\label{eq:BBM1}
  u_t + \left(u + \frac{3}{4}u^2\right)_x - \frac{1}{6}u_{xxt} = 0, \qquad u(x,0) = \eta_0(x).
\end{equation}
We seek a solution of \eqref{eq:gbbm}, \eqref{eq:gbbm2} in the form $\eta=u+v$, where $u\in C([0,\infty),H^{s})$ is the solution of \eqref{eq:BBM1} and $v$ satisfies
\begin{equation}\label{eq:BBMdiff}
  v_t + \left(v + \frac{3}{4}(v^2 + 2uv) - \frac{1}{2}\zeta(v+u)\right)_x -\frac{1}{6}v_{xxt} = -\frac{1}{4}\zeta_{xtt} +\frac{1}{2} \zeta_t, \qquad v(x,0) = 0.
\end{equation}

Let us prove the existence of such $v \in C([0,\infty),H^{s}) $ by induction on $[s]$. First, we assume $[s]=0$. Taking the scalar product of \eqref{eq:BBMdiff} with $v$ in $L^2$, we obtain
\begin{equation*}
  \od{}{t}\int\limits_\R\left(\frac{1}{2} v^2(x,t) +  \frac{1}{12}{v}_x^2(x,t) \right)\ud x = -\int\limits_\R\frac{1}{4}\zeta_{xtt}v\ud x +\int\limits_\R\frac{1}{2}\zeta_{t}v\ud x + \int_{\R}\left(\frac{3}{2}uv - \frac{1}{2}\zeta(v+u)\right)v_x\ud x.
\end{equation*}
Using the Sobolev and H\"older inequalities, we get
\begin{multline}\label{eq:estim2}
  \od{}{t}\|v(\cdot,t)\|^2_1 \le C\bigg(\|\zeta_{tt}(\cdot,t)\|_0 \|v(\cdot,t)\|_1 +\|\zeta_{t}(\cdot,t)\|_0 \|v(\cdot,t)\|_0 + \|u(\cdot,t)\|_0 \|v(\cdot,t)\|^2_1 \\ + \|\zeta(\cdot,t)\|_1 \|v(\cdot,t)\|_1\left(\|v(\cdot,t)\|_0 + \|u(\cdot,t)\|_0 \right)\bigg).
\end{multline}
After integrating \eqref{eq:estim2} on the interval $(0,t)$ we obtain
\begin{multline}\label{eq:estim3}
  \|v(\cdot,t)\|_1^2 \le C\sup\limits_{t\in[0,T]}\|v(\cdot,t)\|_1 \int\limits_0^T \bigg(\|\zeta_{tt}(\cdot,s)\|_0+\|\zeta_{t}(\cdot,s)\|_0 + \|u(\cdot,s)\|_0\|v(\cdot,s)\|_1 + \\ \|\zeta(\cdot,s)\|_1 \left(\|v(\cdot,s)\|_0 + \|u(\cdot,s)\|_0\right)\bigg)\ud s,
\end{multline}
which is valid for any $t\in[0,T]$. Hence, $\sup\limits_{t\in[0,T]}\|v(\cdot,t)\|_1^2$ also can be estimated by the right hand-side of \eqref{eq:estim3}. Dividing by $\sup\limits_{t\in[0,T]} \|v(\cdot,t)\|_1^2$ and applying the Gronwall inequality, we deduce
\begin{multline*}
\sup\limits_{t\in[0,T]}\|v(\cdot,t)\|_1 \le C \int\limits_0^T \left(\|\zeta_{tt}(\cdot,s)\|_0 + \|\zeta_{t}(\cdot,s)\|_0 + \|\zeta(\cdot,s)\|_{1} \|u(\cdot,s)\|_0\right) \ud s \times \\ \exp\left(\int\limits_0^T (\|u(\cdot,s)\|_0 + \|\zeta(\cdot,s)\|_1)\ud s\right).
\end{multline*}
Using this estimation and applying a fixed point argument, we can obtain the existence of the solution in the case $s\in [0,1]$.

By induction, now we assume the existence of $v$ for $[s]<\alpha$ for some integer $\alpha > 1$ and let us prove it for $[s]=\alpha$. To this end, let us take the $\od{^\alpha}{x^\alpha}$, $\alpha \le [s]$ derivative of \eqref{eq:BBMdiff}, multiply the resulting equation by $v_\alpha := \od{^\alpha v}{x^\alpha}$ and integrating in $x$ over $\R$, we obtain
\begin{multline}\label{eq:withal}
\od{}{t}\int\limits_\R\left(\frac{1}{2}v_\alpha^2(x,t) + \frac{1}{12}{v_\alpha}_x^2(x,t)\right)\ud x = - \int\limits_\R\left(\frac{1}{4}\left(\od{^\alpha}{x^\alpha}\zeta_{xtt}\right)v _\alpha\right)\ud x +\int\limits_\R\left(\frac{1}{2}\left(\od{^\alpha}{x^\alpha}\zeta_{t}\right)v _\alpha\right)\ud x+ \\ \int\limits_\R\left(\frac{3}{4}\od{^\alpha}{x^\alpha}(v^2 + 2uv) - \frac{1}{2}\od{^\alpha}{x^\alpha}(\zeta(v+u))\right){v_\alpha}_x\ud x.
\end{multline}
All the terms can be treated as above, except $\int\limits_\R \od{^\alpha}{x^\alpha}(v^2){v _\alpha}_x \ud x$, which is not zero in general. Using the induction hypothesis and the fact that $\alpha - 1 \ge 1$, we can estimate
\begin{equation*}
  \left|\int\limits_\R\od{^\alpha}{x^\alpha}(v^2){v _\alpha}_x\ud x\right|\leq C\|v_\alpha\|_1^2 \|v\|_{\alpha-1} < M\| v_\alpha\|_1^2.
\end{equation*}
Using the last estimation along with \eqref{eq:withal}, the Sobolev and H\"older inequalities, as above, we prove the required estimation for $v_\alpha$, which completes the proof.
\end{proof}

\section{The optimization problem}\label{sec:opt}

In this section we turn to the optimization problem for the gBBM equation \eqref{eq:gbbm}. We assume that the wave making piston is a solid, non-deformable object. Thus, its shape, given by a localized function $\zeta_0(x)$, is preserved during the motion and, accordingly,  it is sufficient to prescribe the trajectory of its barycenter,  $x = x_0(t)$. Consequently, the time-dependent bathymetry is given by the following equation
\begin{equation*}
  h(x,t) = d\ -\ \zeta_0\bigl(x - x_0(t)\bigr).
\end{equation*}

The piston shape $\zeta_0(x)$ and its trajectory $x_0(t)$ will be determined as a solution of the optimization problem. More precisely, in the next section we will find numerically these functions in order to produce the largest possible wave (in the $L_2$ sense) in a given subinterval $\I = [a, b]$ of the numerical wave tank at some fixed time $T > 0$. In other words, we minimize the following functional:
\begin{equation}\label{eq:cost1}
  J(x_0,\zeta_0) = -\int\limits_\I \eta(x,T)^2\ud x \longrightarrow  \min,
\end{equation}
where $\eta(x,t)$ is the solution of \eqref{eq:gbbm}, \eqref{eq:gbbm2}. The existence of this solution is proven in the following
\begin{theorem}\label{thm:opt}
For any constants $\eps, M > 0$, there exists $(x_0^*,\zeta_0^*)\in B_M$ such that
\begin{equation*}
  J(x_0^*,\zeta_0^*) = \inf_{(x_0,\zeta_0) \in B_M} J(x_0,\zeta_0),
\end{equation*}
where $B_M$ is a closed ball in $ H^{2+\eps}[0,T]\times H^{2+\eps}_0([0,1])$, $\eps  > 0$, centered at origin with radius $M$.
\end{theorem}

\begin{proof}
Let $(x_0^n, \zeta_0^n)$ be an arbitrary minimizing sequence of $J$. Since  $ H^{2+\eps}[0,T]\times H^{2+\eps}_0 ([0,1])$ is a Hilbert space, extracting a subsequence, if it is necessary, we can assume that there is $(x_0^*,\zeta_0^*) \in B_M$ such that $(x_0^n, \zeta_0^n) \rightharpoonup (x_0^*,\zeta_0^*)$ weakly  in $B_M$.

Let us denote $\eta^n$ the solution of \eqref{eq:gbbm}, \eqref{eq:gbbm2} with $\zeta = \zeta^n := \zeta_0^n(x - x^n_0(t))$. Let us show that we have $\eta^n(T) \to \eta^*(T)$ in $L^2$, where $\eta^*$ is the solution of \eqref{eq:gbbm}, \eqref{eq:gbbm2} with  $\zeta = \zeta^*$.   Indeed, for $  \teta^{n,m}:=\eta^n-\eta^m$ we have
\begin{equation}\label{eq:optimcauchy}
  \teta^{n,m}_t + \left(\teta ^{n,m}+ \frac{3}{4}\teta^{n,m}(\eta^n + \eta^m) - \frac{1}{2}\zeta^n\teta^{n,m}
  - \frac{1}{2}\zeta^{n,m}\teta^{m} \right)_x - \frac{1}{6}\teta^{n,m}_{xxt} = -\frac{1}{4}\tilde\zeta^{n,m}_{xtt}+\frac{1}{2}\tilde\zeta^{n,m}_{t}.
 \qquad \teta^{n,m} (x,0) = 0.
\end{equation}
Since   $\zeta_{tt}^n \to \zeta_{tt}^* := \partial_{tt} (\zeta_0^*(x - x^\ast_0(t)))$ in $L^2([0,T], L^2)$, multiplying \eqref{eq:optimcauchy} in $L^2$ by $ \teta^{n,m}$, integrating by parts  and applying the Gronwall inequality, we obtain that $\eta^n$  is a Cauchy sequence in $H^1$. Hence, 

\begin{equation*}
  J(x_0^*, \zeta_0^*) = \lim_{n \to \infty} J(x_0^n, \zeta_0^n).
\end{equation*}
This completes the proof of the theorem.
\end{proof}

However, in practice, the functional \eqref{eq:cost1} has to be completed by appropriate constraints in order to provide a solution realizable in practice. For example, the speed of the underwater piston is limited by technological and energy consumption limitations. Some more realistic formulations will be addressed numerically in the next Section.

\section{Numerical results}\label{sec:num}

In order to discretize the gBBM equation \eqref{eq:gbbm}, posed on a finite interval $[\alpha, \beta]$, we use a modern high-order finite volume scheme with the FVCF flux \cite{Ghidaglia2001} and the UNO2 reconstruction \cite{HaOs}. The combination of these numerical ingredients has been extensively tested and validated in the context of the unidirectional wave models \cite{Dutykh2010e} and Boussinesq-type equations \cite{Dutykh2011, Dutykh2011e}. For the time-discretization, we use the third-order Runge-Kutta scheme, which is also used in the \texttt{ode23} function in Matlab \cite{Shampine1997}. In all experiments presented below we assume that the water layer is initially at rest:
\begin{equation*}
  \eta(x,0) = \eta_0(x) \equiv 0.
\end{equation*}
The computational domain $[\alpha, \beta]$ is discretized in $N$ equal subintervals, called usually the control volumes. The time step is chosen locally in order to satisfy the following CFL condition \cite{Courant1928} used in shallow-water models: 
\begin{equation*}
  \Delta t \leq \frac{\Delta x}{\max\limits_{1 \leq i \leq N}{u_i} + \sqrt{gd}}.
\end{equation*}
The values of the physical and numerical parameters used in simulations are given in Table \ref{tab:params}.

On the left and right boundaries we apply the Neumann-type boundary conditions which do not produce reflections. In any case, we stop the simulation before the generated wave reaches the right boundary. We recall that the gBBM equations \eqref{eq:gbbm} describes the unidirectional (rightwards, for instance) wave propagation. So, the influence of the left boundary condition is negligible.

\begin{table}
\centering
\begin{tabular}{lr}
\hline\hline
  \textit{Parameter} & \textit{Value} \\
\hline\hline
  Computational domain $[\alpha, \beta]$: & $[-5, 10]\,\mathsf{m}$ \\
  Wave quality evaluation area $[a, b]$: & $[0, 6]\,\mathsf{m}$ \\
  Number of discretization points $N$: & $1000$ \\
  CFL number: & $1.95$ \\
  Gravity acceleration $g$: & $9.8\,\mathsf{m\,s^{-2}}$ \\
  Undisturbed water depth $d$: & $1.0\,\mathsf{m}$ \\
  Final simulation time $T$: & $6.0\,\mathsf{s}$ \\
  Piston motion total time $T_f$: & $4.0\,\mathsf{s}$ \\
  Piston length $\ell_0$: & $1.0\,\mathsf{m}$ \\
  Piston maximal height $a_0/d$: & $0.12$ \\
  Piston starting point $x_0^0$: & $0.0\,\mathsf{m}$ \\
  Upper bound of the piston position $x_{max}$: & $4.5\,\mathsf{m}$ \\
  Upper bound of the piston speed $v_f$: & $1.5\,\mathsf{m/s}$ \\
  Wave generation limit $x_f$: & $1.0\,\mathsf{m}$ \\  
  $\N$-wave solution ceter $x_m$: & $2.0\,\mathsf{m}$ \\
\hline\hline
\smallskip
\end{tabular}
\caption{\small\em Values of various parameters used in numerical computations.}
\label{tab:params}
\end{table}

Let us describe the constraints that we impose on the shape $\zeta_0(x) \geq 0$ and trajectory $x_0(t)$ of the underwater wavemaker. First of all, we fix the length $2\ell_0$ of this object. Then, we assume that its height is also bounded:
\begin{equation*}
  \max\limits_{x\in\R}\frac{\zeta_0(x)}{d} \leq a_0.
\end{equation*}
We allow the piston to move during the first $T_f \mathsf{s}$. Its motion always starts at the same initial point $x_0^0$ and it is confined to some wave generation area $[x_0(0), x_f] \subseteq [\alpha, \beta]$:
\begin{equation*}
  \supp{x'_0(t)} \subseteq [0, T_f], \qquad x_0(0) = x_0^0, \qquad
  x_0^0 \leq x_0(t) \leq x_f, \quad \forall t\in [0, T_f].
\end{equation*}
However, the cost function $J(x_0, \zeta_0)$ is evaluated at time $T > T_f$ so that the generated waveform can evolve further into the desired shape.

Moreover, we require that the piston speed and acceleration are bounded, since too fast motions are difficult to realize in practice because of the gradually increasing energy consumption:
\begin{equation*}
  \sup_{t \in [0,T_f]}\bigl(|x'_0(t)| + \sqrt{\frac{d}{g}}|x''_0(t)|\bigr) \le v_f
\end{equation*}

In order to parametrize the wavemaker shape, we use only three degrees of freedom $\zeta_0 \bigl(-\frac{\ell_0}{2}\bigr)$, $\zeta_0(0)$, $\zeta_0 \bigl( \frac{\ell_0}{2} \bigr)$ which represent the height of the object in three points equally spaced on the $\supp\zeta_0$. Finally, the continuous shape is reconstructed by applying the interpolation with cubic splines\footnote{Cubic splines ensure that the interpolant belongs to the class $C^2$.} through the following points:
\begin{equation*}
(-\ell_0, 0), \quad \Bigl(-\frac{\ell_0}{2}, \zeta_0\bigl(-\frac{\ell_0}{2}\bigr)\Bigr), \quad
(0, \zeta_0(0)), \quad \Bigl(\frac{\ell_0}{2}, \zeta_0\bigl(\frac{\ell_0}{2}\bigr)\Bigr), \quad (\ell_0, 0).
\end{equation*}
In a similar way we proceed with the parametrization of the piston trajectory $x_0(t)$ which is represented with $4$ degrees of freedom (three in the interior of the interval $(0, T_f)$ and the final point $x_0(T_f)$ which is not fixed as in the case of $\zeta_0$). Obviously, more degrees of freedom can be taken into account when it is needed for a specific application. However, the number of degrees of freedom determines the dimension of the phase space where we seek for the optimal solution. In examples below we operate in a closed subset of $\R^7$. In order to obtain an approximate solution to our constrained optimization problem, we use the function \texttt{fmincon} of the Matlab Optimization Toolbox. This solver is a gradient-based optimization procedure which uses the SQP algorithm. The iterative process is stopped when the default tolerances are met or the maximal number of iterations is reached. In all simulations presented below the convergence of the algorithm has been achieved.

As the first numerical example, we minimize the functional $J(x_0, \zeta_0)$ subject to constraints described above. Basically, this cost function measures the wave deviation from the still water level in a fixed portion $[a, b]$ of the wave tank. Consequently, bigger waves in this interval will provide lower values to the functional $J$. The result of the numerical optimisation procedure is represented on Figure \ref{fig:J1}. The free surface elevation computed at the final time $T$ is shown on Figure \ref{fig:J1}(a). One can see that in the region of interest $[a, b] = [2, 4]$ $\mathsf{m}$ we have a big depression wave which is followed by a wave of elevation. 

Thus, we succeeded to generate a wave suitable for surfing purposes in artificial environments. The computed shape of the underwater object is shown on Figure \ref{fig:J1}(b) and its trajectory is represented on Figure \ref{fig:J1}(c). It is interesting to note that the computed optimal shape is composed of two bumps. The piston trajectory can be conditionally decomposed into three parts. During the first $1.25$ $\mathsf{s}$ we have a stage of slow motion, which is followed by a rapid acceleration and, during the last $0.75$ $\mathsf{s}$, we can observe a backward motion of the piston before it is frozen in its final point. The wave has $T - T_f = 4$ $\mathsf{s}$ to evolve before its \emph{quality} is estimated according the functional $J(x_0, \zeta_0)$.

\begin{figure}
  \centering
  \subfigure[Free surface elevation]%
  {\includegraphics[width=0.69\textwidth]{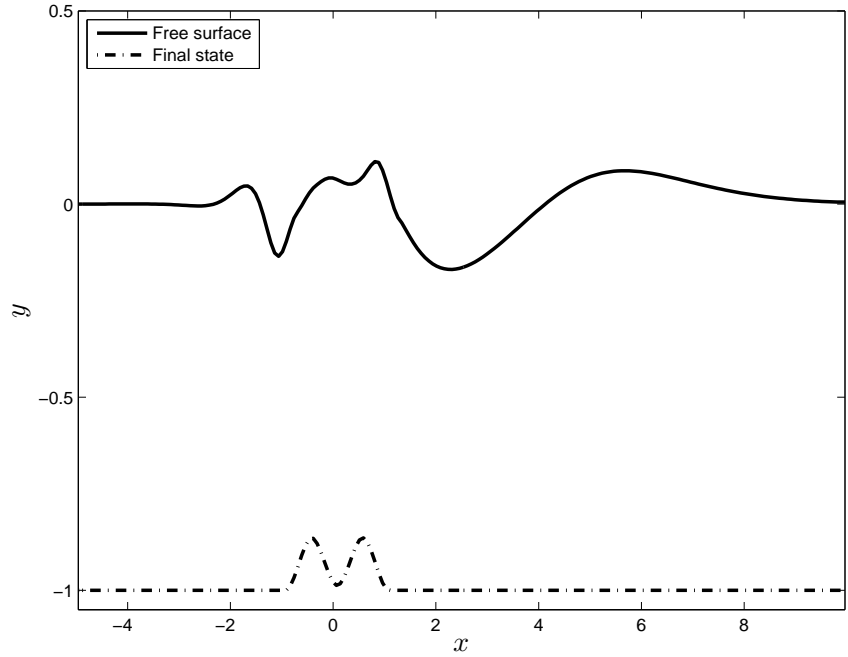}}
  \subfigure[Piston shape]%
  {\includegraphics[width=0.49\textwidth]{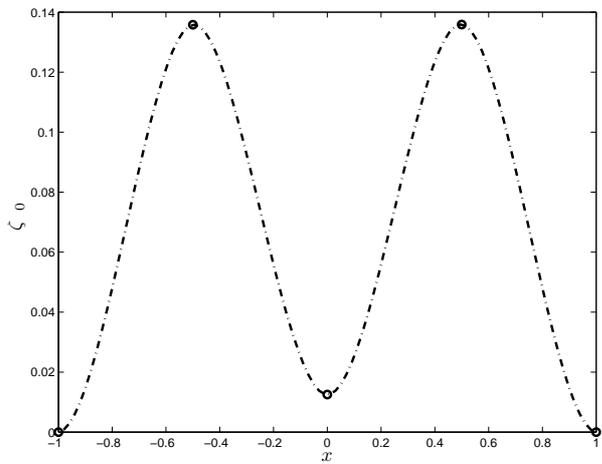}}
  \subfigure[Piston trajectory]%
  {\includegraphics[width=0.49\textwidth]{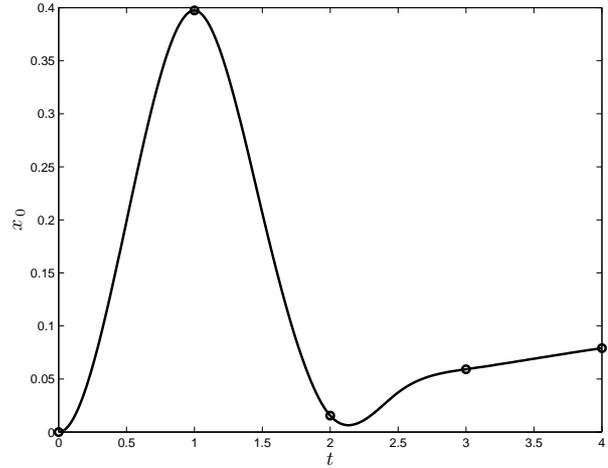}}
  \caption{\small\em Numerical computation of the optimal piston shape and its trajectory for  the functional $J(x_0, \zeta_0)$.}
  \label{fig:J1}
\end{figure}

Since the choice of the functional to minimize is far from being unique, we decided to perform some additional tests. Instead of maximizing the wave height, one can try to maximize, for example, the wave steepness in a given portion of the wave tank. In other words, we will minimize the following functional (subject to the same constraints as above):
\begin{equation*}
  J_1(x_0, \zeta_0) = -\int\limits_\I \eta_x(x,T)\ud x.
\end{equation*}
The result of the numerical optimization procedure is shown on Figure \ref{fig:J2}. One can see on the free surface snapshot \ref{fig:J2}(a) that effectively the wave became steeper. The optimal shape of the wavemaker is almost the same as for the functional $J(x_0, \zeta_0)$. However, the piston trajectory is almost monotonic and close to the uniform motion. This solution might be easier to realize in practice.

\begin{figure}
  \centering
  \subfigure[Free surface elevation]%
  {\includegraphics[width=0.69\textwidth]{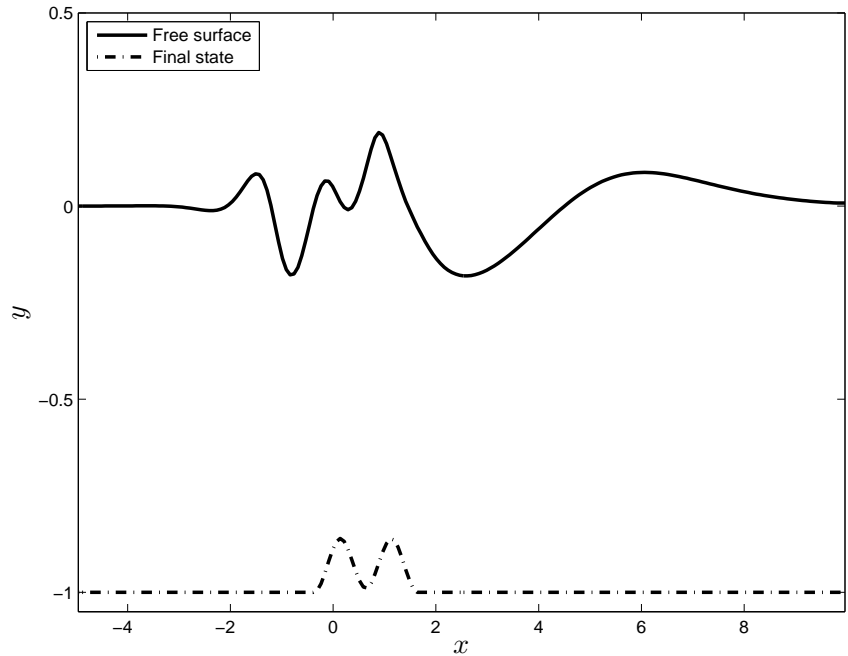}}
  \subfigure[Piston shape]%
  {\includegraphics[width=0.49\textwidth]{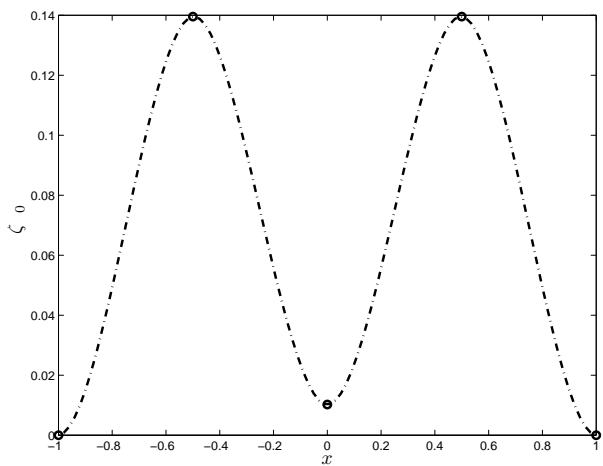}}
  \subfigure[Piston trajectory]%
  {\includegraphics[width=0.49\textwidth]{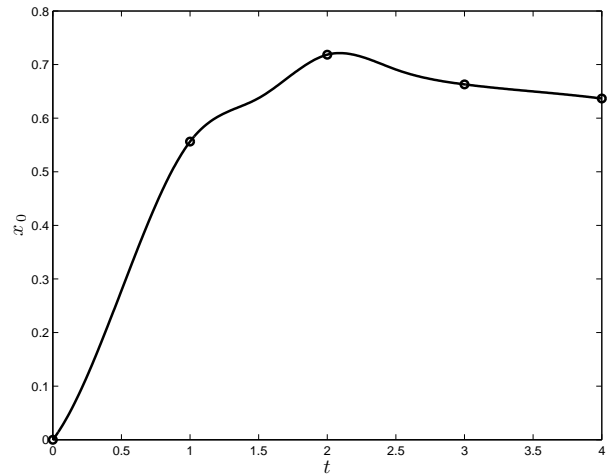}}
  \caption{\small\em Numerical computation of the optimal piston shape and its trajectory for  the functional $J_1(x_0, \zeta_0)$.}
  \label{fig:J2}
\end{figure}

We can also simply minimize the mismatch between the obtained solution and a fixed desired wave profile:
\begin{equation*}
  J_2(x_0, \zeta_0) = \int\limits_\I \left(\eta(x,T) - \eta_T(x)\right)^2\ud x,
\end{equation*}
where $\eta_T(x)$ is a given function on the interval $\I$. 
 
To illustrate this concept, in numerical computations we take the $\N$-wave ansatz put forward by S.~\textsc{Tadepalli} \& C.~\textsc{Synolakis} (1994, 1996) \cite{TS94, Tadepalli1996}:
\begin{equation*}
  \eta_T^{(1)}(x) = (x-x_m)\sech^2(x-x_m), \qquad
  \eta_T^{(2)}(x) = -(x-x_m)\sech^2(x-x_m).
\end{equation*}
The first profile $\eta_T^{(1)}(x)$ corresponds to the leading elevation $\N$-wave solution (LEN), while the second function $\eta_T^{(2)}(x)$ is a typical leading depression $\N$-wave (LDN). The results of optimization procedures are shown on Figures \ref{fig:J3a} and \ref{fig:J3b}. One can notice that the resulting optimal shapes of the wavemaker are completely different (see Figures \ref{fig:J3a}(b) and \ref{fig:J3b}(b)). For the surfing applications the LDN wave might be more interesting. It requires also more uniform piston motion comparing to the LEN wave (see Figures \ref{fig:J3a}(c) and \ref{fig:J3b}(c)).
 
In the final experiment  the target state is the solitary wave for gBBM:
 \begin{equation*}
  \eta_T^{(3)}(x) = 2(c-1)\sech^2( \frac{\sqrt{1-c^{-1} }}{2} |x-x_m|).
  \end{equation*}
As one can notice we can find the shape of the wavemaker which generates waves close to the solitary waves for gBBM  see Figures \ref{fig:J5}).
\begin{figure}
  \centering
  \subfigure[Free surface elevation]%
  {\includegraphics[width=0.69\textwidth]{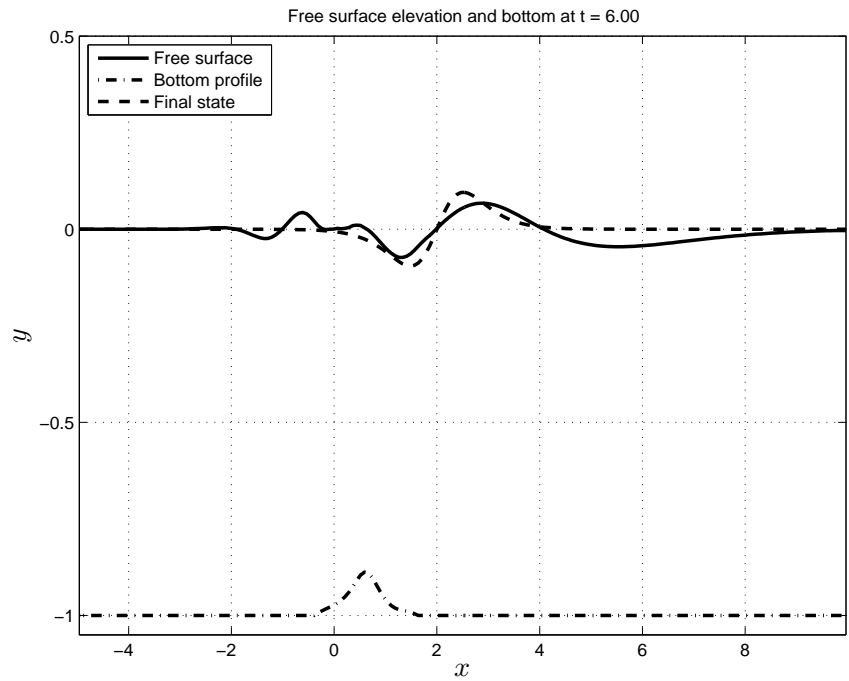}}
  \subfigure[Piston shape]%
  {\includegraphics[width=0.49\textwidth]{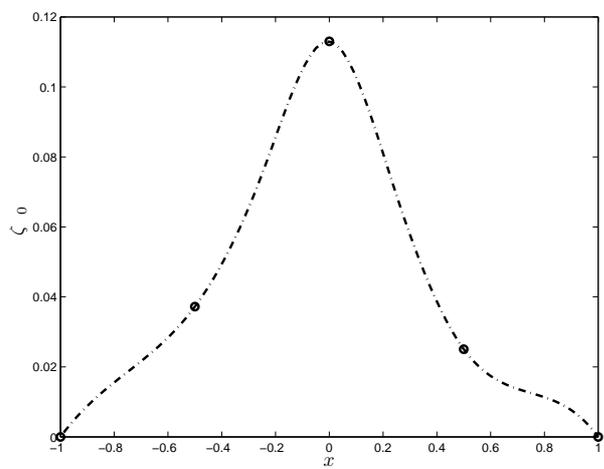}}
  \subfigure[Piston trajectory]%
  {\includegraphics[width=0.49\textwidth]{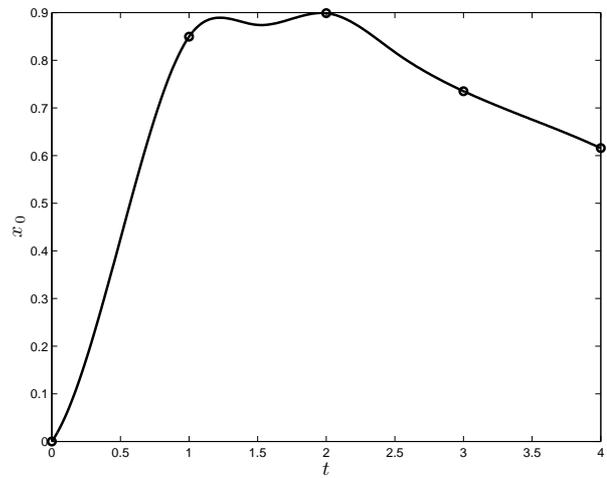}}
  \caption{\small\em Numerical computation of the optimal piston shape and its trajectory for  the functional $J_2(x_0, \zeta_0)$ and the terminal state $\eta_T^{(1)}(x) = (x-x_m)\sech^2(x-x_m)$.}
  \label{fig:J3a}
\end{figure}

\begin{remark}
The arguments used to prove Theorem \ref{thm:opt} can be also applied to show the existence of minimizers for the functionals $J_1(x_0,\zeta_0)$ and $J_2(x_0,\zeta_0)$.
\end{remark}

\begin{figure}
  \centering
  \subfigure[Free surface elevation]%
  {\includegraphics[width=0.69\textwidth]{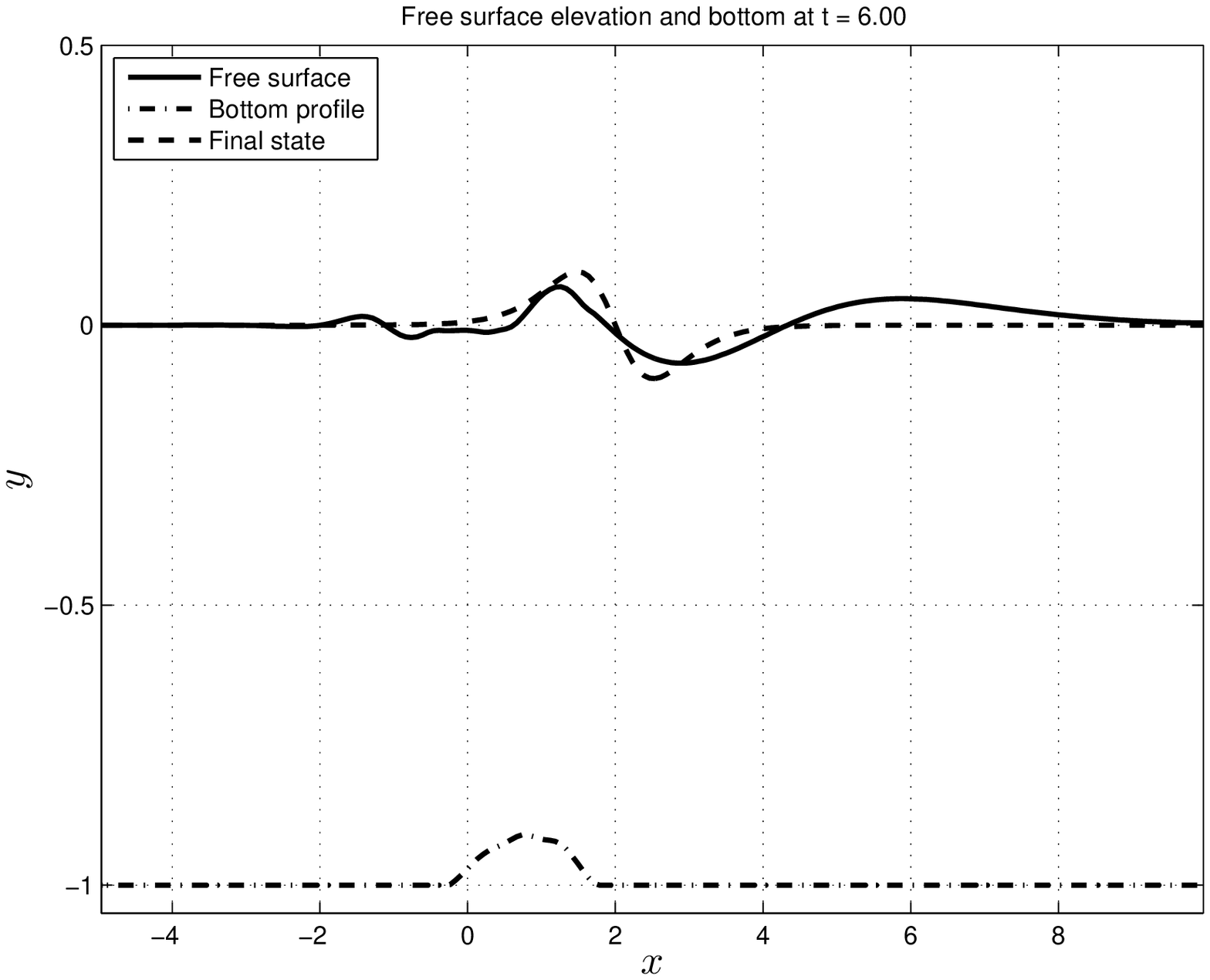}}
  \subfigure[Piston shape]%
  {\includegraphics[width=0.49\textwidth]{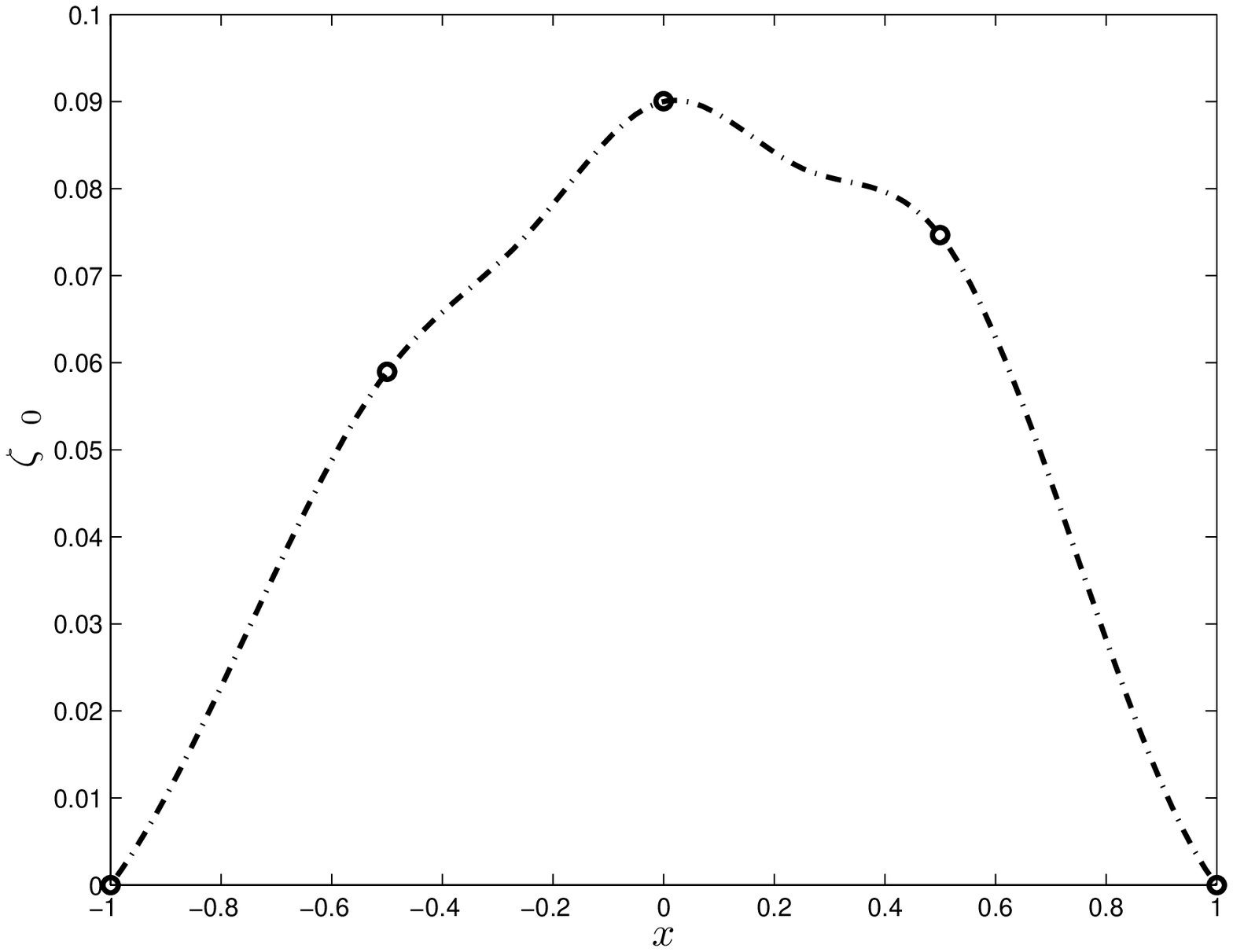}}
  \subfigure[Piston trajectory]%
  {\includegraphics[width=0.49\textwidth]{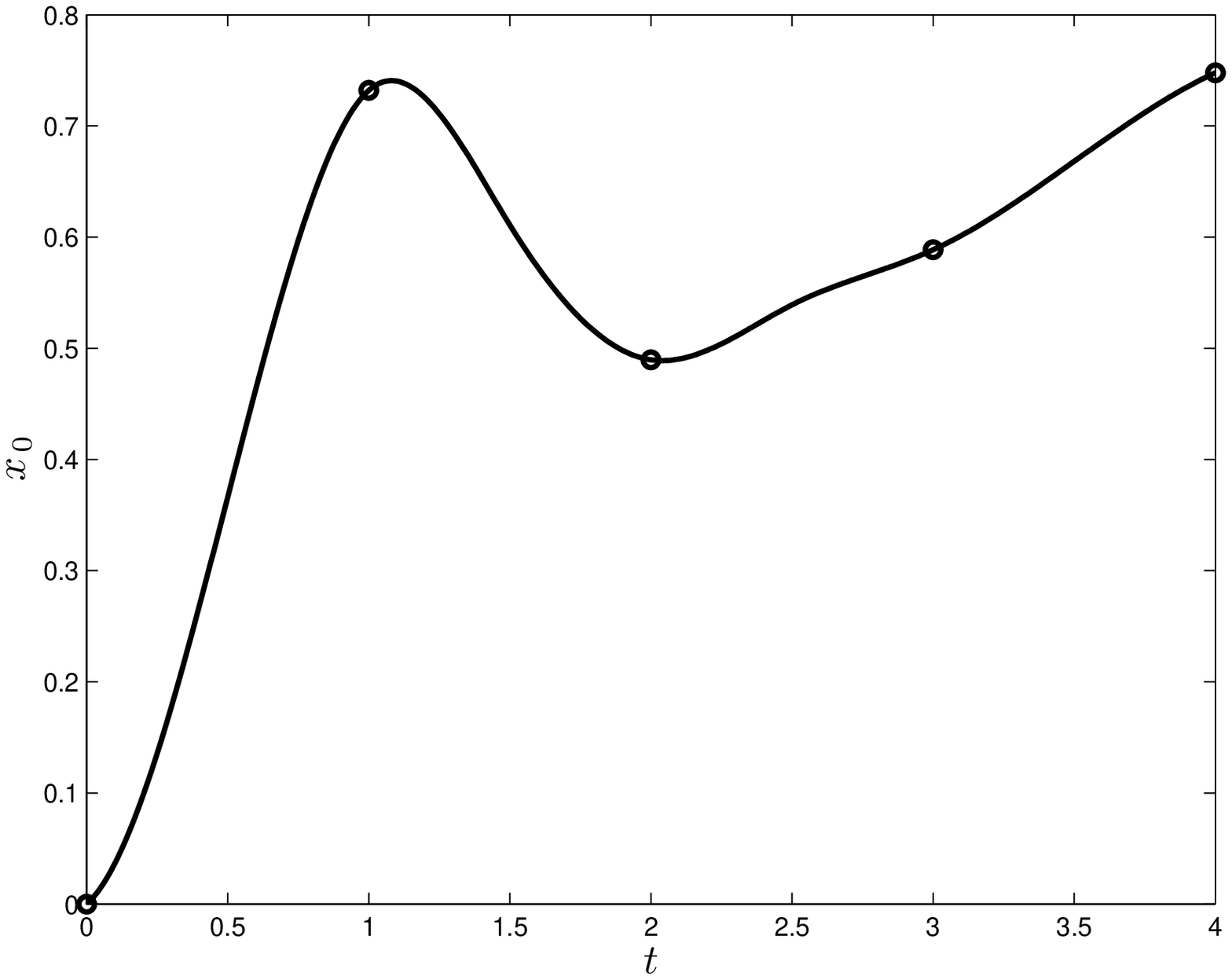}}
  \caption{\small\em Numerical computation of the optimal piston shape and its trajectory for  the functional $J_2(x_0, \zeta_0)$ and the terminal state $\eta_T^{(1)}(x) = -(x-x_m)\sech^2(x-x_m)$.}
  \label{fig:J3b}
\end{figure}

\begin{figure}
  \centering
  \subfigure[Free surface elevation]%
  {\includegraphics[width=0.69\textwidth]{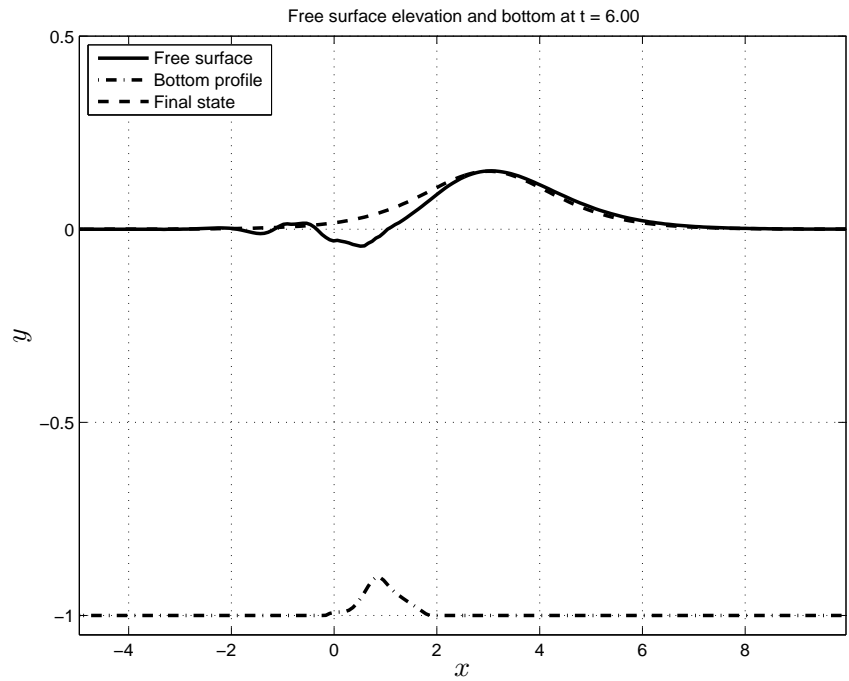}}
  \subfigure[Piston shape]%
  {\includegraphics[width=0.49\textwidth]{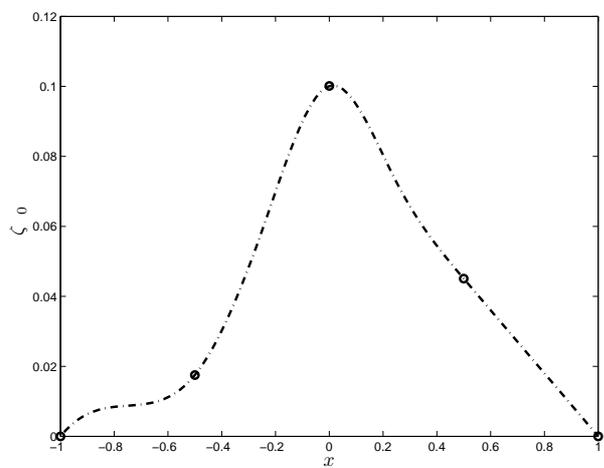}}
  \subfigure[Piston trajectory]%
  {\includegraphics[width=0.49\textwidth]{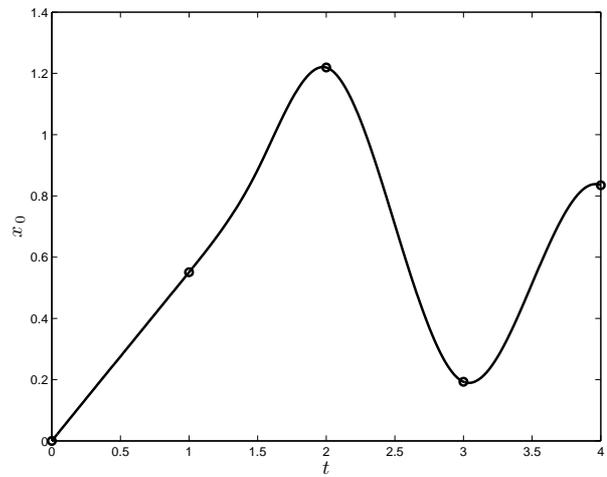}}
  \caption{\small\em Numerical computation of the optimal piston shape and its trajectory for  the functional $J_2(x_0, \zeta_0)$ and the terminal state $\eta_T^{(2)}(x) = (c-1)\sech^2(\frac{\sqrt{1-c^{-1} }}{2} |x-x_m|)$.}
  \label{fig:J5}
\end{figure}

\section{Conclusions}\label{sec:concl}

In the present work we considered the water wave generation problem by disturbances moving along the bottom. This problem has many important applications going even to the design of artificial surfing facilities \cite{InstantSport2012}. In order to study the formation of water waves due to the motion of the underwater piston, we derived a generalized forced BBM (gBBM) equation. The existence and uniqueness of its solutions were rigorously established. The trajectory of the piston was determined as the solution of a thoroughly formulated optimization problem. The existence of minimizers was also proven. Finally, the theoretical developments of this study were illustrated with numerical examples where we solve several constrained optimization problems with various forms of the cost functional. The resulting solutions were compared and discussed.

In future studies this problem will be addressed in the context of more complete bidirectional wave propagation models of Boussinesq-type \cite{BC, DMS2, Mitsotakis2007, DMII}. The optimization algorithm can be also further improved by evaluating the gradients analytically, for example. From the physical point of view, one may want to include some weak dissipative effects for more realistic wave description \cite{Dutykh2007}.
 
\section*{Acknowledgements}

H.~\textsc{Nersisyan} and E.~\textsc{Zuazua} were supported by the project ERC -- AdG FP7-246775 NUMERIWAVES, the Grant PI2010-04 of the Basque Government, the ESF Research Networking Program OPTPDE and Grant  MTM2011-29306 of the MINECO, Spain. D.~\textsc{Dutykh} acknowledges the support from ERC under the research project ERC-2011-AdG 290562-MULTIWAVE. Also he would like to acknowledge the hospitality and support of the Basque Center for Applied Mathematics (BCAM) during his visits.
This work was finished while the third author was visiting the ``Laboratoire Jacques Louis Lions'' with the support of the program ``Research in Paris''.

\bibliography{biblio}
\bibliographystyle{abbrv}

\end{document}